\documentclass[a4paper,12pt]{article}
\usepackage{amsmath}
\usepackage{amsthm}
\usepackage{amsfonts}
\usepackage{mathrsfs}
\usepackage{graphicx}
\usepackage{hyperref}
\newtheorem{theorem}{Theorem}
\newtheorem{lemma}{Lemma}

\numberwithin{equation}{section}

\hypersetup{ pdftitle={Draft},
pdfauthor={Fiki T. Akbar, Bobby E. Gunara}, % tuliskan nama
pdfkeywords={supersymmetry, local existence}, % tuliskan kata kunci
bookmarksnumbered, pdfstartview={FitH}, urlcolor=blue, }

\begin{document}
\pagestyle{plain}

%%%%%%%%%%%%%%%%%%%%%%%%%%%%%%%%%%%%%%%%%%%%%%%%%%%%%%%%%%%%%%%%%%%%%%%%%%%%%%%%%%%%%%%%%%%%

%% TITLE %%

%%%%%%%%%%%%%%%%%%%%%%%%%%%%%%%%%%%%%%%%%%%%%%%%%%%%%%%%%%%%%%%%%%%%%%%%%%%%%%%%%%%%%%%%%%%%

\title{\LARGE\textbf{Justification of the Lugiato-Lefever model from a damped driven $\phi^4$ equation}}

\author{{Fiki T. Akbar$^{\sharp}$, Bobby E. Gunara}$^{\sharp}$, Hadi Susanto$^{\flat}$, \\ \\
$^{\sharp}$\textit{\small Theoretical Physics Laboratory}\\
\textit{\small Theoretical High Energy Physics and Instrumentation Research Group,}\\
\textit{\small Faculty of Mathematics and Natural Sciences,}\\
\textit{\small Institut Teknologi Bandung}\\
\textit{\small Jl. Ganesha no. 10 Bandung, Indonesia, 40132}\\
{\small and} \\
$^{\flat}$\textit{\small Department of Mathematical Sciences, }\\
\textit{\small University of Essex,}\\
\textit{\small Colchester, CO4 3SQ, United Kingdom}\\
\small email: ftakbar@fi.itb.ac.id, bobby@fi.itb.ac.id, hsusanto@essex.ac.uk}

\date{}

\maketitle

%%%%%%%%%%%%%%%%%%%%%%%%%%%%%%%%%%%%%%%%%%%%%%%%%%%%%%%%%%%%%%%%%%%%%%%%%%%%%%%%%%%%%%%%%%%%

%% Abstract %%

%%%%%%%%%%%%%%%%%%%%%%%%%%%%%%%%%%%%%%%%%%%%%%%%%%%%%%%%%%%%%%%%%%%%%%%%%%%%%%%%%%%%%%%%%%%%

\begin{abstract}
The Lugiato-Lefever equation is a damped and driven version of the well-known nonlinear Schr\"odinger equation. It is a mathematical model describing complex phenomena in dissipative and nonlinear optical cavities. Within the last two decades, the equation has gained a wide attention as it becomes the basic model describing optical frequency combs. Recent works derive the Lugiato-Lefever equation from a class of damped driven $\phi^4$ equations closed to resonance. In this paper, we provide a justification of the envelope approximation. From the analysis point of view, the result is novel  and non-trivial as the drive yields a perturbation term that is not square integrable. The main approach proposed in this work is to decompose the solutions into a combination of the background and the integrable component. This paper is the first part of a two-manuscript series.

\end{abstract}

%%%%%%%%%%%%%%%%%%%%%%%%%%%%%%%%%%%%%%%%%%%%%%%%%%%%%%%%%%%%%%%%%%%%%%%%%%%%%%%%%%%%%%%%%%%%

% % I. Introduction % %

%%%%%%%%%%%%%%%%%%%%%%%%%%%%%%%%%%%%%%%%%%%%%%%%%%%%%%%%%%%%%%%%%%%%%%%%%%%%%%%%%%%%%%%%%%%%

\section{Introduction}
\label{sec1}

The Lugiato-Lefever equation is given by \cite{lugi87} 
\begin{equation}
\mathrm{i} A_{\tau} = - A_{\xi\xi}  - \frac{\mathrm{i}\alpha}{2} A - \frac{3\lambda}{2\omega} |A|^2 A + F,\quad \xi\in\mathbb{R},\,\tau\geq0, \label{NLSeq}
\end{equation}
which is nothing else but a damped driven nonlinear Schr\"odinger equation. It models spatiotemporal pattern formation in dissipative, diffractive and nonlinear optical cavities submitted to a continuous laser pump.  The same model was shown rather immediately to also appear in dispersive optical ring cavities \cite{hael92}. Lugiato-Lefever equation has raised a wide interest particularly following its recent succesful experimental application in the study of broadband microresonator-based optical frequency combs \cite{delh07,kipp11}, that has opened applicative avenues (see \cite{lugi15,chem17} for reviews on the subject).

Recently Ferr\'e et al.\ \cite{ferr17} showed that the dynamics of the Lugiato-Lefever equation can also be obtained from a driven dissipative sine-Gordon model. The former equation is a single envelope approximation, i.e., modulation equation, of the latter. Even in the region far from the conservative limit, where the approximation is expectedly no longer valid, they were reported to still exhibit qualitatively similar dynamical behaviors. 

Herein, instead of the sine-Gordon equation, we consider a nonlinear damped driven $\phi^4$ model 
\begin{equation}
u_{tt} + \epsilon^2 \alpha u_{t} - \beta u_{xx} + \gamma u - \lambda u^3 = \epsilon^3 h\left(e^{\mathrm{i}\Omega t} + e^{-\mathrm{i}\Omega t}\right), \label{NKGeq}
\end{equation}
where $\alpha, \beta, \gamma > 0$ and $\epsilon$ is a small positive parameter. The nonlinearity is considered to be 'softening', i.e., $\lambda < 0$. The 'hardening' case $\lambda > 0$ will be discussed in the second part of this paper series, whose results can be extended to the sine-Gordon equation. Introducing the slow time and space $\tau$ and $\xi$ defined as $\tau = \epsilon^2 t$ and $\xi = \epsilon\sqrt{\frac{2\omega^3}{\gamma \beta}}\left(x - vt\right)$, where $v = d\omega/dk = \beta k/\omega$ is a group velocity of the linear traveling wave and $k$ and $\omega$ are satisfying dispersion relation $\omega^2 = \beta k^2 + \gamma$, we define the slowly modulated ansatz function as
\begin{equation}
X(t,x) = \epsilon A(\tau,\xi)e^{\mathrm{i}(kx - \omega t)} + \frac{\lambda\epsilon^3}{9\beta k^2 - 9\omega^2 + \gamma}A(\tau,\xi)^3 e^{3\mathrm{i}(kx - \omega t)} + \mathrm{c.c}\:. \label{AnsatzFunction}
\end{equation}
The modulation amplitude $A$ is a complex valued function satisfying Eq.\ \eqref{NLSeq}, where $F(\tau,\xi) = - \frac{h}{2\omega} e^{-\mathrm{i}(\kappa \xi - \nu\tau)}$ with $\kappa = k/\epsilon$ and $\Omega = \gamma/\omega - \epsilon^2 \nu$. Inserting the ansatz function (\ref{AnsatzFunction}) into (\ref{NKGeq}), we get the residual terms
\begin{equation}
\label{residualterm}
\mathrm{Res}(t) = \mathcal{O}(\epsilon^4).
\end{equation}

The same modulation equation has been derived in early reports, e.g., in \cite{mora74,kaup78,noza84} to describe matters driven by an external ac field. Analytical studies of various solutions of the damped, driven continuous nonlinear Schr\"odinger equation have also been reported in \cite{24,5}. Nevertheless, despite the long history of the problem, a rigorous justification of the approximation is interestingly still lacking. The main challange is due to the external drive $F$ that is not integrable. Our work presented in this paper is to provide an answer to the missing piece.

An early work justifying the modulation equation without damping and drive was due to \cite{14,9}. The presence of external drives would not bring any problem should one consider nonlinear systems that correspond to a parabolic linear operator  \cite{14,coll90,hart91,schn94,schn94b}. In the context of Eq.\ \eqref{NKGeq}, this corresponds to $\alpha\to\infty$, in which case the modulation equation would be a Ginzburg-Landau-type equation, i.e., there is no factor $\mathrm{i}$ on the left hand side of \eqref{NLSeq}.

Recently we consider the reduction of a Klein-Gordon equation with external damping and drive into a damped driven discrete nonlinear Schr\"odinger equation \cite{muda18}. To overcome the nonintegrability of the solutions, we worked in a periodic domain. The present report extends our result in \cite{muda18} by proposing a method that works also in $L^2(\mathbb{R})$. 

This paper is organized as follows. To provide a rigorous proof justifying the modulation equation, we formulate our method in Section \ref{sec2} by decomposing the solutions into the background and particular parts. Using a small-amplitude appeoximation, we derive the Lugiato-Lefever equation from Section \ref{sec3} presents the local and global existence of homogeneous solutions of the amplitude equation. The main result on the error bound of the  approximation as time evolves is presented in Section \ref{sec4}.

\section{Solution decomposition}
\label{sec2}

Since the external drive term $F(\tau,\xi)$ is not integrable in the spatial variable, i.e., $F(\tau,\xi) \notin L^{p}(\mathbb{R})$ for any integer $ 1 \leq p < \infty$, in general $A(\tau,\xi)$ is also not integrable. Let $A_p(\tau,\xi)$ be a particular solution of equation (\ref{NLSeq}) which can be written as
\begin{equation}
A_p(\tau,\xi) = R\;e^{-\mathrm{i}(\kappa \xi - \nu\tau)}\:, \label{partsoln}
\end{equation}
where $R$ is a complex constant such that 
\begin{equation}
R  = -\frac{h}{2\omega}\;\frac{1}{\frac{3\lambda}{2\omega}|R|^2  - (\kappa^2 + \nu) +\frac{\mathrm{i}\alpha}{2}}\:,
\end{equation}
and $|R|^2$ satisfies cubic equation, 
\begin{equation*}
\frac{9\lambda^2}{4\omega^2}|R|^6 - \frac{6\lambda}{2\omega}(\kappa^2 + \nu) |R|^4 + \left[\frac{\alpha^2}{4} + (\kappa^2 + \nu)^2\right] |R|^2 - \frac{h^2}{4\omega^2} = 0\:.
\end{equation*}
Since $\lambda<0$, the cubic equation has only one real solution. Furthermore, we have
\begin{equation}
\|A_{p}(\tau,\xi)\|_{L^{\infty}(\mathbb{R})} = \sup_{\xi \in \mathbb{R}}\; |A_{p}(\tau,\xi)| = |R| \leq +\infty\:.
\end{equation}

To handle the non-integrability condition we can work on $\mathbb{T}$, rather than on $\mathbb{R}$ (see \cite{muda18} where a similar problem was considered in the discrete case). In this paper, we propose a different approach by introducing the decomposition, 
\begin{eqnarray}
\label{decomposition}
A(\tau,\xi) & := & e^{\mathrm{i}\nu \tau}\phi(\tau,\xi) + A_{p}(\tau,\xi) \nonumber \\
& = & e^{\mathrm{i}\nu \tau}\left[\phi(\tau,\xi) + \eta(\xi)\right],
\end{eqnarray}
where $\phi(\tau,\xi)$ is the integrable term % $A_{p}(\tau,\xi)$ is the particular solution (the non integrable-term) given by equation (\ref{partsoln}) 
and $\eta(\xi) = R e^{-\mathrm{i}\kappa \xi}$. The initial condition for our system is
\begin{equation}
A(0,\xi) = \varphi(\xi) + \eta(\xi)\:.
\label{tamb}
\end{equation}
with $\varphi \in H^{k}(\mathbb{R})$. %This approach is inspired by work of Banica and Vega in \cite{banica08}.
Here, the space $H^k(\mathbb{R})$ with $k$ is a nonnegative integer denotes the Sobolev space with norm defined as
\begin{equation}
\|\phi\|_{H^{k}(\mathbb{R})} = \left[\sum_{i=0}^{k} \|D^{i}_{\xi}\phi\|^2_{L^2(\mathbb{R})}\right]^{1/2}\:,
\end{equation}
and $H^0(\mathbb{R}) = L^2(\mathbb{R})$ with
\begin{equation}
\|\phi\|_{L^{2}(\mathbb{R})} =\left[\int_{\mathbb{R}}\;|\phi(\xi)|^2\:d\xi\right]^{1/2}\:.
\end{equation}

The differential equation for $\phi$ is given by,
\begin{equation}
\mathrm{i} \phi_{\tau} = - \phi_{\xi\xi} - \frac{\mathrm{i}\alpha}{2}\phi - \left( \frac{3\lambda}{2\omega}|R|^2 - \nu\right)\phi + N(\phi)  \:, 
\label{NLSHomoDE}
\end{equation}
where the nonlinearity term is given by,
\begin{equation}
N(\phi) := - \frac{3\lambda}{2\omega}\left[|\phi + \eta|^2 - |\eta|^2\right](\phi + \eta)\:. \label{NonlinearPart}
\end{equation}
Using \eqref{tamb}, we have the initial condition 
\[\phi(0,\cdot) = \varphi.\]

\section{Local and Global Existence of the inhomogeneous Nonlinear Schr\"odinger equation}
\label{sec3}

In this section, we prove the local and global existence for the inhomogeneous part of the nonlinear Schr\"odinger equation. For an excellent review, an interested reader can further consult, for example, \cite{bourgain99,cazenave}. 

The local existence is stated in the following theorem.
\begin{theorem}
	\label{localsoln}
	
	Let $k \geq 1$ be an integer. For every $\varphi \in H^k(\mathbb{R})$, there exists a positive constant $\tau_m$ depending on the initial data and $k$ such that the differential equation (\ref{NLSHomoDE}) admits an unique maximal solution $\phi(\tau)$ on $[0,\tau_m)$ with
	\begin{equation}
	\phi \in C\left([0,\tau_m), H^k(\mathbb{R})\right)\:,
	\end{equation}
	and either, 
	\begin{enumerate}
		\item $\tau_m = +\infty$, and (\ref{NLSHomoDE}) admits a global solution, or
		\item $\tau_m < +\infty$, then $\|\phi(\tau)\|_{H^k(\mathbb{R})} \rightarrow \infty$ as $\tau \rightarrow \tau_m$ and the solution blow up at finite time $\tau_m$. Moreover, $\limsup \|\phi(\tau)\|_{L^{\infty}(\mathbb{R})}  \rightarrow \infty$ as $\tau \rightarrow \tau_m$.
	\end{enumerate}
\end{theorem}

\begin{proof}
	
	We prove the theorem in three steps.

	\textbf{Step 1. Local existence}. Using Duhamel's formula, we can write the solution of the differential equation (\ref{NLSHomoDE}) as
	\begin{equation}
	\label{integralEq}
	\phi(\tau) = U(\tau)\phi_{0} + \mathrm{i}\int_{0}^{\tau}\:U(\tau-\tau')\left[\frac{\mathrm{i}\alpha}{2}\phi(\tau') + \left( \frac{3\lambda}{2\omega}|R|^2 - \nu\right)\phi(\tau') - N(\phi(\tau'))\right]\:d\tau' \:,
	\end{equation}
	where $U(\tau)$ is one dimensional free Schr\"odinger time evolution operator. Define
	\begin{equation}
	\mathcal{B} = \left\{\phi\in C\left([0,\tilde{\tau}],H^k(\mathbb{R})\right) \left| \|\phi(\tau)\|_{H^k(\mathbb{R})} \leq M,\; \forall \tau \in [0,\tilde{\tau}) \right.\right\}
	\end{equation}
	be a Banach space equipped with norm,
	\begin{equation}
	\|\phi\|_{\mathcal{B}} = \sup_{\tau \in [0,\tilde{\tau})}\;\|\phi(\tau)\|_{H^k(\mathbb{R})}\:.
	\end{equation}
	For $\phi \in H^k(\mathbb{R})$, we define a nonlinear operator
	\begin{equation}
	\mathcal{K}[\phi](\tau) = U(\tau)\phi_0 +  \mathrm{i}\int_{0}^{\tau}\:U(\tau-\tau')\left[\frac{\mathrm{i}\alpha}{2}\phi(\tau') + \left( \frac{3\lambda}{2\omega}|R|^2 - \nu\right)\phi(\tau') - N(\phi(\tau'))\right]\:d\tau'.
	\end{equation}
	We want to prove that the operator $\mathcal{K}$ is a contraction mapping on $\mathcal{B}$.
	Using the fact that the free Schr\"odinger operator $U(\tau)$ is a linear operator and unitary on $H^k(\mathbb{R})$, we have
	\begin{equation}
	\|U(\tau)\phi\|_{H^k(\mathbb{R})} = \|\phi\|_{H^k(\mathbb{R})}\:,
	\end{equation}
	for any $\phi \in H^k(\mathbb{R})$, thus we get
	\footnotesize
	\begin{eqnarray}
	\|\mathcal{K}[\phi](\tau)\|_{H^k(\mathbb{R})} & \leq & \|U(\tau)\phi_0\|_{H^k(\mathbb{R})} + \int_{0}^{\tau}\:\left\|U(\tau-\tau')\left[\frac{\mathrm{i}\alpha}{2}\phi(\tau') + \left( \frac{3\lambda}{2\omega}|R|^2 - \nu\right)\phi(\tau') - N(\phi(\tau'))\right]\right\|_{H^k(\mathbb{R})}\:d\tau' \nonumber \\
	& \leq & \|\phi_0\|_{H^k(\mathbb{R})} + \int_{0}^{\tau}\:\left\|\left[\frac{\mathrm{i}\alpha}{2}\phi(\tau') + \left( \frac{3\lambda}{2\omega}|R|^2 - \nu\right)\phi(\tau') - N(\phi(\tau'))\right]\right\|_{H^k(\mathbb{R})}\:d\tau'\nonumber\\
	& \leq & \|\phi_0\|_{H^k(\mathbb{R})} + \left(\frac{\alpha}{2} + \frac{3|\lambda|}{2\omega}|R|^2 + |\nu|\right)\tau \delta + \tau \sup_{\tau' \in [0,\tau]} \;\|N(\phi(\tau'))\|_{H^k(\mathbb{R})}
	\end{eqnarray}
	\normalsize
	Note that,
	\begin{equation}
	\left[|\phi + \eta|^2 - |\eta|^2\right](\phi + \eta) = |\phi|^2\phi + 2|\phi|^2 \eta + |\eta|^2\phi + \eta^2 \bar{\phi} + \phi^2 \bar{\eta}\:.
	\end{equation}
	Since $H^{k}(\mathbb{R})$ is an algebra for $k > 1/2$, thus for $\tau \in [0,\tilde{\tau})$ we get
	\begin{eqnarray}
	\|N(\phi(\tau))\|_{H^k(\mathbb{R})} & \leq & \frac{3|\lambda|}{2\omega}\left(\|\phi\|^3_{H^k(\mathbb{R})} + 3|R|\|\phi\|^2_{H^k(\mathbb{R})} + 2|R|^2\|\phi\|_{H^k(\mathbb{R})}\right) \nonumber\\
	& \leq & \frac{3|\lambda|M}{2\omega}\left(M^2 + 3|R|M + 2|R|^2\right)\:.
	\end{eqnarray}
	Assuming that $\|\phi_0\|_{H^k(\mathbb{R})} < \delta/2$, hence we have,
	\begin{equation}
	\|\mathcal{K}[\phi](\tau)\|_{\mathcal{B}} \leq \frac{M}{2} + \frac{\tilde{\tau} M}{2}\left[\alpha + 2|\nu| +  \frac{3|\lambda|}{\omega}\left(3|R|^2 + 3|R|M + M^2\right)\right]\:.
	\end{equation}
	If we pick,
	\begin{equation}
	\tilde{\tau} < \frac{1}{\alpha + 2|\nu| + \frac{3|\lambda|}{\omega}\left(3|R|^2 + 3|R|M + M^2\right)}\:,
	\end{equation}
	then $\mathcal{K}$ map $\mathcal{B}$ to itself.

	Let $\phi,\varphi \in \mathcal{B}$, then
	{\footnotesize
	\begin{eqnarray}
	\|\mathcal{K}[\phi](\tau) - \mathcal{K}[\varphi](\tau)\|_{H^k(\mathbb{R})} & \leq & \int_{0}^{\tau}\:\left\|\left(\frac{\mathrm{i}\alpha}{2} + \frac{3\lambda}{2\omega}|R|^2 - \nu\right)(\phi(\tau') - \varphi(\tau')) - (N(\phi(\tau')) - N(\varphi(\tau')))\right\|_{H^k(\mathbb{R})}\:d\tau' \nonumber\\
	& \leq & \int_{0}^{\tau}\:\left(\frac{\alpha}{2} + \frac{3|\lambda|}{2\omega}|R|^2 + |\nu|\right)\left\|(\phi(\tau') - \varphi(\tau')) \right\|_{H^s(\mathbb{R})}\;d\tau' \nonumber \\
	& & \qquad + \int_{0}^{\tau}\:\|N(\phi(\tau')) - N(\varphi(\tau'))\|_{H^k(\mathbb{R})}\:d\tau' \nonumber\\
	& \leq & \tau \left(\frac{\alpha}{2} + \frac{3|\lambda|}{2\omega}|R|^2 + |\nu|\right) \sup_{\tau' \in [0,\tau]} \left\|(\phi(\tau') - \varphi(\tau')) \right\|_{H^k(\mathbb{R})} \nonumber \\
	&  & \qquad + \tau \sup_{\tau' \in [0,\tau]}\|N(\phi(\tau')) - N(\varphi(\tau'))\|_{H^k(\mathbb{R})}
	\end{eqnarray}
	}
	Note we have the following inequalities,
	\begin{eqnarray}
	|\phi^2 - \varphi^2| & = & |\phi + \varphi||\phi - \varphi| \leq (|\phi| + |\varphi|)|\phi - \varphi|, \nonumber\\
	||\phi|^2 - |\varphi|^2| & = &|\phi\bar{\phi} - \varphi\bar{\varphi}| = |\phi(\bar{\phi} - \bar{\varphi}) + \bar{\varphi}(\phi - \varphi)| \leq (|\phi| + |\varphi|)|\phi - \varphi|, \nonumber \\
	||\phi|^2\phi - |\varphi|^2\varphi| & = & |(|\phi|^2 + |\varphi|^2)(\phi-\varphi) + \phi\varphi(\bar{\phi}-\bar{\varphi})| \leq \frac{3}{2}(|\phi|^2 + |\varphi|^2)|\phi-\varphi|. \nonumber
	\end{eqnarray}
	Thus, for $\tau \in [0,\tilde{\tau}]$ we get
	\begin{equation}
	\|N(\phi(\tau)) - N(\varphi(\tau))\|_{H^k(\mathbb{R})} \leq  \frac{3|\lambda|}{2\omega}\left(2|R|^2 + 6|R|M + 3M^2\right)\|\phi(\tau) - \varphi(\tau)\|_{H^k(\mathbb{R})}.
	\end{equation}
	Hence, we have
	\begin{equation}
	\|\mathcal{K}[\phi](\tau) - \mathcal{K}[\varphi](\tau)\|_{\mathcal{B}}  \leq \frac{\tilde{\tau}}{2}\left[\alpha + 2|\nu| +  \frac{9|\lambda|}{\omega}\left(|R|^2 + 2|R|M + M^2\right)\right]\|\phi(\tau) - \varphi(\tau)\|_{\mathcal{B}}
	\end{equation}
	If we pick,
	\begin{equation}
	\tilde{\tau} < \min\left[\frac{1}{\alpha + 2|\nu| +  \frac{3|\lambda|}{\omega}\left(3|R|^2 + 3|R|M + M^2\right)}, \frac{2}{\alpha + 2|\nu| + \frac{9|\lambda|}{\omega}\left(|R|^2 + 2|R|M + M^2\right)}\right]\:,
	\end{equation}
	then the nonlinear operator $\mathcal{K}$ is a contraction mapping in $\mathcal{B}$. Therefore by Banach fixed point theorem, there exist a fixed point of $\mathcal{K}$ which is a solution of (\ref{integralEq}), hence of (\ref{NLSHomoDE}).

	\textbf{Step 2. Uniqueness} Let $\phi,\tilde{\phi} \in C\left([0,\tilde{\tau}), H^{k}(\mathbb{R}) \right)$ are solutions of (\ref{NLSHomoDE}) with the same initial condition $\varphi \in H^{k}(\mathbb{R})$. Let $\eta = \tilde{\phi} - \phi$. By Duhamel's formula, we have
	\begin{equation}
	\eta(\tau) = \mathrm{i}\int_{0}^{\tau}\:U(\tau-\tau')\left[\frac{\mathrm{i}\alpha}{2}\eta(\tau') + \left( \frac{3\lambda}{2\omega}|R|^2 - \nu\right)\eta(\tau') - \left(N(\tilde{\phi}(\tau')) - N(\phi(\tau'))\right)\right]\:d\tau'\:.
	\end{equation}
	Taking norm in $H^{k}(\mathbb{R})$ and using the property of the nonlinear term as in the local existence proof, we get
	\begin{equation}
	\|\eta(\tau)\|_{H^{k}(\mathbb{R})} \leq C(M)\int_0^{\tau}\:\|\eta(\tau')\|_{H^{k}(\mathbb{R})}\;d\tau'\:.
	\end{equation}
	By Gronwall's inequality, we conclude that $\|\eta(\tau)\|_{H^{k}(\mathbb{R})} = 0$ for all $\tau \in [0,\tilde{\tau})$, hence $\phi' = \phi$.
	
	\textbf{Step 3. Maximal solution.} We can construct the maximal solution by repeating the step 1 with the initial condition $\phi(\tilde{\tau}-\tau_0)$ for some $0< \tau_0 <\tilde{\tau}$ and using the uniqueness condition to glue the solutions. Clearly, if $\tau_m = +\infty$, then we have a global solution and if $\tau_m < +\infty$, then $\|\phi(\tau)\|_{H^k(\mathbb{R})} \rightarrow \infty$ as $\tau \rightarrow \tau_m$.  
	
	Finally, we will show that if $\tau_m < +\infty$, then $\limsup \|\phi_{\tau}\|_{L^\infty(\mathbb{R})} \rightarrow \infty$ as $\tau \rightarrow \tau_m$. Suppose that $\limsup \|\phi_{\tau}\|_{L^\infty(\mathbb{R})} < \infty$ as $\tau \rightarrow \tau_m$. Since $\phi \in C\left([0,\tau_m), H^{k}(\mathbb{R}) \right)$, and $H^k$ is embedded to $L^{\infty}$ for $k\geq1$, then
	\begin{equation*}
	\sup_{\tau \in [0,\tau_m)} \|\phi(\tau)\|_{L^\infty(\mathbb{R})} \leq \sup_{\tau \in [0,\tau_m)} \|\phi(\tau)\|_{H^k(\mathbb{R})} \leq M\:.
	\end{equation*}
	Using Duhamel's formula using the property of the nonlinear term as in the local existence proof, we get
	\begin{equation}
	\|\phi(\tau)\|_{H^k(\mathbb{R})} \leq \|\varphi\|_{H^k(\mathbb{R})} + C(M)\int_{0}^{\tau}\:\|\phi(\tau')\|_{H^k(\mathbb{R})}\;d\tau'\:.
	\end{equation}
	Applying Gronwall's inequality, then for $\tau \in [0,\tau_m)$, we have
	\begin{equation}
	\|\phi(\tau)\|_{H^k(\mathbb{R})} \leq \|\varphi\|_{H^k(\mathbb{R})}e^{C(M)\tau_m}\:,
	\end{equation}
	which contradicts with the blow up of $\|\phi(\tau)\|_{H^k(\mathbb{R})}$ at $\tau \rightarrow \tau_m$. Hence, $\limsup \|\phi_{\tau}\|_{L^\infty(\mathbb{R})} \rightarrow \infty$ as $\tau \rightarrow \tau_m$.  
\end{proof}

Due to the type of nonlinearity and presence of the damping term, differential equation (\ref{NLSHomoDE}) does not possess any conserved quantities. However, we can define the energy function associated to equation (\ref{NLSHomoDE}) as
\begin{equation}
E[\phi](\tau) = \int_{\mathbb{R}}\;\left\{|\phi_{\xi}|^2 - \left( \frac{3\lambda}{2\omega}|R|^2 - \nu\right)|\phi|^2 - \frac{3\lambda}{4\omega}\left[|\phi + \eta|^2 - |\eta|^2\right]^2  \right\}\;d\xi\:.
\label{energyfunction}
\end{equation}
Since $\lambda < 0$, then $E$ is non-negative. It is worth mentioning that if $\alpha = 0$ (no damping term), then this energy function $E$ is conserved. Furthermore, we still can use this energy function to prove the global existence in $H^{1}(\mathbb{R})$ for small initial energy and prove that the solution is in fact possess more regularity and prove the global existence on $H^{k}(\mathbb{R})$.

Now, we prove the following lemma about energy estimate
\begin{lemma}
	Let $\phi$ be a solution of differential equation (\ref{NLSHomoDE}) with initial data $\varphi \in H^{1}(\mathbb{R})$ such that $E_0 = E[\varphi] \leq \delta$ with $\delta$ is a positive real constant. There exist a positive real constant $\delta_0$ such that for every $0 < \delta < \delta_0$ we have the following estimate,
	\begin{equation}
	E[\phi](\tau) \leq K e^{-\alpha\tau}\:, \label{energydecay}
	\end{equation} 
	where $K$ is a real positive constant depending on the initial data.
\end{lemma}

\begin{proof}
	First, we multiply equation (\ref{NLSHomoDE}) with $-\bar{\phi}_{\tau} - \frac{\alpha}{2} \bar{\phi}$, integrate over the spatial variable $\xi$ and keep the real part. Then, we get
	\begin{equation}
	\frac{d E[\phi](\tau)}{d\tau} + \alpha E[\phi](\tau) = \frac{3\alpha\lambda}{4\omega} \int_{\mathbb{R}}\:\left[|\phi + \eta|^2 - |\eta|^2\right]|\phi|^2\;d\xi\;,
	\end{equation}
	Since $\alpha >0$, using Cauchy-Schwartz we get
	\begin{eqnarray}
	\frac{3\alpha\lambda}{4\omega} \int_{\mathbb{R}}\:\left[|\phi + \eta|^2 - |\eta|^2\right]|\phi|^2\;d\xi & \leq & \alpha \left[\frac{3|\lambda|}{4\omega} \int_{\mathbb{R}}\:\left[|\phi + \eta|^2 - |\eta|^2\right]^2\;d\xi\right]^{\frac{1}{2}} \left[\frac{3|\lambda|}{4\omega} \int_{\mathbb{R}}\:|\phi|^4 \;d\xi\right]^{\frac{1}{2}} \nonumber \\
	& \leq & \alpha E^{\frac{1}{2}}\left\|\frac{3|\lambda|}{4\omega}\phi\right\|^{2}_{L^{4}(\mathbb{R})}\:.
	\end{eqnarray}
	Using the one dimensional Gagliardo-Nirenberg-Sobolev inequality,
	\begin{equation}
	\left\|\frac{3|\lambda|}{4\omega}\phi\right\|_{L^{4}(\mathbb{R})} \leq C\|\phi_{\xi}\|^{1/4}_{L^{2}(\mathbb{R})}\left\|\frac{3|\lambda|}{4\omega}\phi\right\|^{3/4}_{L^{2}(\mathbb{R})}
	\end{equation}
	and Young inequality,
	\begin{equation}
	a b \leq p\; a^{\frac{1}{p}} + (1-p)b^{\frac{1}{1-p}}
	\end{equation}
	for $a,b >0$ and $p \in (0,1)$, we get
	\begin{equation}
	\left\|\frac{3|\lambda|}{4\omega}\phi\right\|_{L^{4}(\mathbb{R})} \leq C\left[\|\phi_{\xi}\|^{2}_{L^{2}(\mathbb{R})} + \frac{3|\lambda|}{2\omega}|R|^2 \|\phi\|^2_{L^{2}(\mathbb{R})} \right]^{1/2} \leq C\;E^{1/2}.
	\end{equation}
	Thus, we have
	\begin{equation}
	\frac{d E(\tau,\phi)}{d\tau} + \alpha E(\tau,\phi) \leq \alpha C E^{3/2}.
	\end{equation}
	Integrating this inequality, we get
	\begin{equation}
	E(\tau) \leq \frac{4 E_0}{\left[C\sqrt{E_0}(1 - e^{\alpha \tau/2}) + 2\;e^{\alpha \tau/2}\right]^2}\:. 
	\end{equation}
	Pick $\delta_0 = 4/C^2$, then we have the desired inequality (\ref{energydecay}).
	
\end{proof}

Now, we prove the global existence in the following theorem.
\begin{theorem}
	\label{globalexistence}
	Let $k \geq 1$ be an integer. For every $\varphi \in H^k(\mathbb{R})$ such that $E_0 = E[\varphi] \leq \delta$ with $\delta$ is a positive real constant. There exists a positive real constant $\delta_0$ such that for every $0 < \delta < \delta_0$, the differential equation (\ref{NLSHomoDE}) admits a unique global solution $\phi$ which belongs to
	\begin{equation}
	\phi \in C\left([0,+\infty), H^k(\mathbb{R})\right)\:.
	\end{equation}
	
\end{theorem}

\begin{proof}
	
	First consider the case of $k=1$. Pick $\delta_0$ as in the previous theorem, then the global existence is directly follow from (\ref{energydecay}). Thus, we have
	\begin{equation}
	\label{GEH1}
	\phi \in C\left([0,+\infty), H^1(\mathbb{R})\right)\:.
	\end{equation}

	Now consider the case of $k>1$. In Theorem \ref{localsoln}, we already construct a unique local maximal solution such that,
	\begin{equation}
	\phi \in C\left([0,\tau^k_m), H^k(\mathbb{R})\right)\:.
	\end{equation}
	We need to prove that $\tau^k_m = +\infty$. Consider $\tau_0 <+\infty$, then we have
	\begin{equation}
	\sup_{\tau \in [0,\tau_0]} \|\phi(\tau)\|_{H^{1}(\mathbb{R})} < \infty\:.
	\end{equation}
	Since $H^{1}(\mathbb{R}) \hookrightarrow L^{\infty}(\mathbb{R})$, then (\ref{GEH1}) implies
	\begin{equation}
	\sup_{\tau \in [0,\tau_0]} \|\phi(\tau)\|_{L^{\infty}(\mathbb{R})} < \infty\:.
	\end{equation}
	Applying the blow up alternative in the local existence theorem, we deduce that $\tau_m^k > \tau_0$. Since $\tau_0$ is arbitrary, then we conclude $\tau_m^k = +\infty$ and the proof is finished.
	
\end{proof}

\section{Main Result}
\label{sec4}

Before we state the main result, we will prove a lemma about the bound of the leading approximation function and the residual function. 

\begin{lemma}
	\label{lemmaA}
	For every $A(0) = \varphi + \eta$, where $\varphi \in H^{k}(\mathbb{R})$ with integer $k > 4$ such that $E[\varphi]$ is small, then there exists a positive real constant $C_X$ and $C_R$ that depend on $\|A(0)\|_{L^\infty(\mathbb{R})}$ such that
	\begin{eqnarray}
	\|X_t(t)\|_{L^\infty(\mathbb{R})} + \|X(t)\|_{L^\infty(\mathbb{R})} & \leq & \epsilon\;C_X,\nonumber\\
	\|\mathrm{Res}(t)\|_{L^\infty(\mathbb{R})} & \leq & \epsilon^4 \:C_R\:,
	\end{eqnarray}
	for all $t \in [0,+\infty)$. Furthermore, we also have $X(t) \in C_b^{k-1}(\mathbb{R})$ and $X_t(t) \in C_b^{k-3}(\mathbb{R})$ for all $t \in [0,+\infty)$.  
\end{lemma}

\begin{proof}
	
	From Theorem \ref{globalexistence}, we have $\phi \in C\left([0,+\infty), H^k(\mathbb{R})\right)$ for integer $k\geq 1$. Since $H^{k}(\mathbb{R})$ is embedded into $L^{\infty}(\mathbb{R})$ for $k \geq 1$, then using decomposition (\ref{decomposition}) and the fact that $A_{p}(\tau) \in L^{\infty}(\mathbb{R})$, we get
	\begin{equation}
	\|A(\tau)\|_{L^{\infty}(\mathbb{R})} \leq C_A.
	\end{equation}
	
	Since $L^{\infty}(\mathbb{R})$ is Banach algebra with respect to pointwise multiplication, then we can estimate equation (\ref{AnsatzFunction}),
	\begin{equation*}
	\|X(t)\|_{L^{\infty}(\mathbb{R})} \leq \epsilon \;C_{1}.
	\end{equation*}
	Since $k > 4$, $\|\phi_{\xi\xi}\|_{H^{k}(\mathbb{R})} \leq C\|\phi\|_{H^{k-2}(\mathbb{R})}$ then from equation (\ref{NLSHomoDE}), we get
	\begin{equation*}
	\|\phi_{\tau}\|_{H^k(\mathbb{R})} \leq C\:.
	\end{equation*}
	Thus, using decomposition (\ref{decomposition}), we have estimate for the first derivative of $A$ with respect to $\tau$ and conclude that
	\begin{equation*}
	\|X_{t}(t)\|_{L^{\infty}(\mathbb{R})} \leq \epsilon\;C_{2}\:,
	\end{equation*} 
	hence proved the first inequality.
	
	The residual terms consist of power of $A$ and higher derivative of $A$ up to second order (both time and space). Since $k > 4$, then using Sobolev embedding and equation (\ref{NLSHomoDE}), we get the bound for second derivative of $\phi$ with respect to both space and time. Then, using (\ref{decomposition}), we proved the second inequality.
	
	For the second part of the theorem, note that $A_p(\tau,\xi)$ is smooth and bounded function. Since $\phi(\tau) \in H^k(\mathbb{R})$ and $\phi_{\tau}(\tau) \in H^{k-2}(\mathbb{R})$ and $k>4$, then by Sobolev embedding we have
	\begin{equation*}
	\|\phi(\tau)\|_{C_b^m(\mathbb{R})} \leq C \|f\|_{H^k(\mathbb{R})} \:,
	\end{equation*} 
	for $k > m + 1$, hence we proved the theorem.
\end{proof}

We define the error term by writing $\epsilon^2 y(t) = u(t) + X(t)$, with $X(t)$ is a leading approximation term and $u(t)$ is the exact solution of our original equation. The evolution equation for the error term is given by
\begin{equation}
\label{errorEq}
\begin{cases}
y_{tt} + \epsilon^2 \alpha y_{t} - \beta y_{xx} + \gamma y - \lambda (\epsilon^4 y^3 + 3 X^2 y + 3 \epsilon^2 X y^2) + \epsilon^{-2}\mathrm{Res}(t) = 0,\\
y(0,\cdot) = f, \\
y_t(0,\cdot) = \epsilon g.
\end{cases}
\end{equation}
Since we expect $\epsilon$ is small, then we assume that $\gamma > \frac{\epsilon^2\alpha}{2}$. We prove that the error function $y(t)$ remains bounded over time. First, we can convert the differential equation (\ref{errorEq}) into integral equation \cite{ficken57},
\begin{equation}
\label{errorIntEq}
y(t,x) = \Phi[y](t,x) = A[f](t,x) + B[\epsilon g](t,x) + M\left[\lambda (y^3 + 3X^2 y + 3Xy^2) + \mathrm{Res}\right],
\end{equation}
where $A$, $B$ and $M$ are integral operators defined as
\begin{eqnarray*}
	A[f](t,x) & = & \frac{1}{2}e^{-\hat{\alpha}t} \left[f(x+t) + f(x-t) + \hat{\alpha} \int_{x-t}^{x+t}\:f(z)J_0(\epsilon w \zeta_0)\:dz + \int_{x-t}^{x+t}\:f(z)\frac{\partial J_0(\epsilon w \zeta_0)}{\partial t}\:dz \right], \\
	B[\epsilon g](t,x) & = & \frac{\epsilon}{2}e^{-\hat{\alpha}t}\int_{x-t}^{x+t}\:g(z)J_0(\epsilon w \zeta_0)\:dz, \\
	M[h](t,x) & = & -\frac{1}{2} \int_{0}^{t}\:\int_{x-t+s}^{x+t-s}\:e^{-\hat{\alpha}(t-s)}h(z,s)J_0(\epsilon w \zeta)\:dz\;ds,
\end{eqnarray*}
with $\hat{\alpha} = \epsilon^2\alpha/2$, $J_0$ is a zeroth order Bessel function, and
\begin{eqnarray*}
	\epsilon w & = & \sqrt{\gamma^2 - \hat{\alpha}^2},\\
	\zeta^2 & = & (t-s)^2 - (x-z)^2, \\
	\zeta_0^2 & = & t^2 - (x-z)^2.
\end{eqnarray*}
Since the intervals of integration are $s \in [0,t]$ and $z \in [x-t+s,x+t-s]$, then $\zeta^2 \geq 0$, thus we can set $\zeta \geq 0$.

Note that,
\begin{equation}
|J_n(z)| \leq \frac{1}{\Gamma(1+n)}\left(\frac{|z|}{2}\right)^n e^{\mathrm{Im}(z)}\:.
\end{equation}  
Since $w\zeta \in \mathbb{R}$, then we have,
\begin{eqnarray}
|J_0(\epsilon w\zeta)| & \leq & 1, \nonumber \\
\left|\frac{\partial J_0}{\partial t}(\epsilon w\zeta) \right| = \left|-\frac{\epsilon w}{\zeta}(t-s)J_1(\epsilon w\zeta)\right| & \leq & \frac{\epsilon^2 w^2}{2} |t-s|\:.
\end{eqnarray}
Using this two estimates, now we can estimate the integral operators
\begin{eqnarray}
\label{IntegralEstimate}
|A[f](t,x)| & \leq & e^{-\hat{\alpha}t}\|f\|_{L^{\infty}(\mathbb{R})}\left(1 + \hat{\alpha}t + \frac{\epsilon^2 w^2 t^2}{2}\right), \nonumber \\
|B[\epsilon g](t,x)| & \leq &  \epsilon e^{-\hat{\alpha}t} t \|g\|_{L^{\infty}(\mathbb{R})}, \nonumber \\
|M[h](t,x)| & \leq & \int_{0}^{t}\: e^{-\hat{\alpha}(t-s)} (t-s)\|h(s)\|_{L^{\infty}(\mathbb{R})}\;ds.
\end{eqnarray}
We assume that $\|f\|_{L^{\infty}(\mathbb{R})},\|g\|_{L^{\infty}(\mathbb{R})} \leq C_0$. Banach algebra property of ${L^{\infty}(\mathbb{R})}$ enables us to bound the nonlinear term for each $D>0$ and for all $\|y\|_{L^{\infty}(\mathbb{R})} \leq \;D$ and then we have 
\begin{eqnarray}
\|y(t)\|_{L^{\infty}(\mathbb{R})} & \leq &  \left[\left(1 + \hat{\alpha}t + \epsilon t + \frac{\epsilon^2 w^2 t^2}{2}\right) C_0 + \frac{\epsilon^2 t^2}{2} C_R + \frac{|\lambda| \epsilon^2 t^2}{2}\left( \epsilon^4 D^3 + \epsilon^3 C_X D^2\right) \right]\nonumber\\
& &\qquad + |\lambda|\epsilon^2 C_{X}^{2}\int_{0}^{t}\:(t-s)\|y(s)\|_{L^{\infty}(\mathbb{R})}\;ds\:,
\end{eqnarray}
where we already used the fact that $e^{-\hat{\alpha}t} \leq 1$ for $t \geq 0$. If $\epsilon > 0$ is sufficiently small, i.e. $\epsilon \in (0,\epsilon_0)$ with $\epsilon_0$ is a positive real constant, then for each $D>0$ and for every $\|y(t)\|_{L^{\infty}(\mathbb{R})} \leq D$, we can find a positive real  constant $M$ independent of $\epsilon$ such that,
\begin{equation*}
\frac{1}{2}\left(|\lambda|\epsilon^4 D^3 + |\lambda|\epsilon^3 C_X D^2\right) < M\:.
\end{equation*}
Thus, as long as we have $\|y(t)\|_{L^{\infty}(\mathbb{R})}$ staying in the ball of radius $D$, we have
\begin{equation*}
\|y(t)\|_{L^{\infty}(\mathbb{R})} \leq a(t) + \int_{0}^{t}\;b(s)\|y(s)\|_{L^{\infty}(\mathbb{R})} \;ds\;,
\end{equation*}
where 
\begin{equation}
a(t) = \left[\left(1 + \hat{\alpha}t + \epsilon t + \frac{\epsilon^2 w^2 t^2}{2}\right) C_0 +  \frac{\epsilon^2 t^2}{2} C_R + \epsilon^2 M t^2 \right]\;,\qquad b(s) = |\lambda|\epsilon^2 C_{X}^{2} (t-s)\:.
\end{equation}
The function $a(t)$ is continuous and non-decreasing and $b(t)$ is positive for $t \in [0,T_0/\epsilon]$. Then applying Gronwall's inequality we get
\begin{equation}
\|y(t)\|_{L^{\infty}(\mathbb{R})} \leq a(t) e^{|\lambda| C_X^2 \epsilon^2 t^2/2}\:.
\end{equation}
Therefore,
\begin{equation}
\|y(t)\|_{L^{\infty}(\mathbb{R})} \leq \left[\left(1 + \frac{\epsilon \alpha T_0}{2} + T_0 + \frac{w^2 T_0^2}{2}\right) C_0 + \frac{T_0^2}{2} C_R +  M T_0^2 \right] e^{|\lambda| C_X^2 T_0^2/2}\:.
\end{equation}

Let $C_y = \left(1  + T_0 + \frac{w^2 T_0^2}{2}\right) C_0 + C M T_0^2$ and $D = C_y e^{|\lambda| C_X^2 T_0^2/2}$ and make $\epsilon_0$ smaller such that $\epsilon < \frac{2M}{\alpha}T_0$, then we have
\begin{equation}
\|y(t)\|_{L^{\infty}(\mathbb{R})} \leq D\;,
\end{equation} 
for $t \in [0,T_0/\epsilon]$. Hence we proved the following theorem,
\begin{theorem}
	\label{MainTheorem}
	Let $A(\tau,\xi)$ is the solution of equation (\ref{NLSeq}) such that $A \in C^2\left([0,T_1],C_b^k(\mathbb{R})\right)$ for integer $k \geq 0$ and $X$ is leading approximation function (\ref{AnsatzFunction}). Let $u(t,x)$ be a solution of equation (\ref{NKGeq}). Then for each $T_0 < T_1$ and each $C_0 > 0$, there exist $\epsilon_0$ and $D>0$ such that for every $\epsilon \in (0,\epsilon_0)$ with
	\begin{equation*}
	\|u(0,\cdot) - X(0,\cdot)\|_{L^{\infty}(\mathbb{R})} \leq \epsilon^2 C_0,\qquad \|u(0,\cdot) - X(0,\cdot)\|_{L^{\infty}(\mathbb{R})} \leq \epsilon^3 C_0\;,
	\end{equation*}
	then the following inequality 
	\begin{equation*}
	\|u(t,\cdot) - X(t,\cdot)\|_{L^{\infty}(\mathbb{R})} \leq \epsilon^2 D
	\end{equation*}
	holds for $t \in [0,T_0/\epsilon]$.  
\end{theorem}

\section*{Acknowledgement}

The work of FTA is partly supported by Ministry of Research, Technology and Higher Education of the Republic of Indonesia through  PDUPT 2018.  FTA would like to thank The Abdus Salam ICTP for associateship in 2018. The authors, BEG and HS, acknowledge Ministry of Research, Technology and Higher Education of the Republic of Indonesia for partial financial supports through the World Class Professor 2018 program.

\end{document}